\documentclass[manuscript,nonacm]{acmart}
\acmConference[]{}{}{}
\acmBooktitle{}

\usepackage{url}
\usepackage{ascmac}
\usepackage{fancybox}
\usepackage{amsthm}
\usepackage{color}
\usepackage{bm}
\usepackage{threeparttable}
\usepackage{enumitem}
\usepackage{framed}

\newcommand{\qer}[2]{\mathcal{E}_{#1}(#2)}
\newcommand{\fcer}[2]{\mathcal{C}_{#1}^{1}(#2)}
\newcommand{\scer}[2]{\mathcal{C}_{#1}^{2}(#2)}

\newcommand{\recba}{\mbox{\sf RBA}}
\newcommand{\GBA}{\mbox{\sf GBA}}
\newcommand{\OptGBA}{\frac 12 \mbox{\sf-GBA}}
\newcommand{\EpsGBA}{(\frac 12 -\varepsilon)\mbox{\sf-GBA}}

\newcommand{\ignore}[1]{}

\newcommand{\sig}[1]{\langle #1 \rangle}
\newcommand{\Echo}{\mathsf{echo}}
\newcommand{\FirstVote}{\mathsf{vote\text{-}1}}
\newcommand{\SecondVote}{\mathsf{vote\text{-}2}}

\newcommand{\EchoCert}{\mathcal{E}}
\newcommand{\Cert}{\mathcal{C}}

\newcommand{\ThirdVote}{\mathsf{vote\text{-}3}}

\newcommand{\bos}[1]{\medskip\noindent\textbf{#1}}
 
\begin{document}

\title{Optimal Communication Complexity of Authenticated Byzantine Agreement}

% you can include author information in the source, but `anonymous' option will hide it

\author{Atsuki Momose}
\email{momose@sqlab.jp}
\affiliation{%
  \institution{Nagoya University}
  \country{}
}

\author{Ling Ren}
\email{renling@illinois.edu}
\affiliation{%
  \institution{University of Illinois at Urbana-Champaign}
  \country{}
}

\begin{abstract}
Byzantine Agreement (BA) is one of the most fundamental problems in distributed computing, and its communication complexity is an important efficiency metric.
It is well known that quadratic communication is necessary for BA in the worst case due to a lower bound by Dolev and Reischuk.
This lower bound has been shown to be tight for the unauthenticated setting with $f < n/3$ by Berman et al. but a considerable gap remains for the authenticated setting with $n/3 \le f < n/2$.

This paper provides two results towards closing this gap.
Both protocols have a quadratic communication complexity and have different trade-offs in resilience and assumptions.
The first protocol achieves the optimal resilience of $f < n/2$ but requires a trusted setup for threshold signature. 
The second protocol achieves near optimal resilience $f \le (1/2 - \varepsilon)n$ in the standard PKI model.
\end{abstract}

% keywords, ACM classification and conference information can be omitted for submission

\maketitle

\section{Introduction}
Byzantine Agreement (BA) is one of the most fundamental problems in distributed algorithms~\cite{lamport1982byzantine}.
It also serves as an important building block in cryptography and distributed systems. 
At a high level, Byzantine agreement is the problem for $n$ parties to agree on a value, despite that up to $f$ of them may behave arbitrarily (called Byzantine faults). 
Arguably the most important efficiency metric of Byzantine Agreement is the communication complexity, since communication will be the bottleneck in applications like state machine replication and cryptocurrency. 

Dolev and Reischuk proved that a quadratic number of messages are necessary for deterministic BA protocols~\cite{dolev1985bounds}.
More formally, they showed that even in the authenticated setting (i.e., with a public key infrastructure and ideal digital signatures), any BA protocol has at least one execution where quadratic number of messages are sent by honest parties.
The tightness of this lower bound was partially established by Berman et al. in the unauthenticated setting with $f < n/3$.
However, for decades, the best known protocol for the authenticated setting (with $f \ge n/3$) remains the classic Dolev-Strong protocol~\cite{dolev1983authenticated}~\footnote{Dolev-Strong solves a related problem called Byzantine broadast, but it is easy to transform it into a BA protocol.}, which uses quadratic messages but cubic communication.
The reason is that in Dolev-Strong, the messages can contain up to $f+1$ signatures.
Therefore, the optimal communication complexity of authenticated BA with $f \ge n/3$ has been an open problem for a very long time.

This paper provides two results that help close this gap. More specifically, we show the following two theorems. 
Note that when $f \ge n/3$, it is necessary to adopt the synchronous and authenticated setting. Under asynchrony~\cite{fischer1985impossibility}, partial synchrony~\cite{dwork1988consensus}, or the unauthenticated setting~\cite{fischer1986easy}, BA is impossible for $f \ge n/3$.

    \begin{theorem} \label{theo:ba_exist_optimal}
        Assuming a threshold signature scheme, there exists a Byzantine agreement protocol with $O(\kappa n^2)$ communication complexity tolerating $f < n/2$ faults where $n$ is the number of parties and $\kappa$ is a security parameter.
    \end{theorem}
    
    \begin{theorem} \label{theo:ba_without_trusted_setup}
        Assuming a public-key infrastructure, there exists a Byzantine agreement protocol with $O(\kappa n^2)$ communication complexity tolerating $f \le (\frac 12 - \varepsilon)n$ faults where $n$ is the number of parties, $\kappa$ is a security parameter, and $\varepsilon$ is any positive constant. 
    \end{theorem}    

\begin{table*}[tb]
  \centering
  \begin{threeparttable}
  \begin{tabular}{|c|c|c|c|} \hline
    \textsf{protocol} & \textsf{model} & \textsf{communication} & \textsf{resilience} \\ \hline \hline
    Berman et al.~\cite{berman1992bit} & unauthenticated & $O(n^2)$ & $f < n/3$ \\
    Dolev-Strong~\cite{dolev1983authenticated} & PKI & $O(\kappa n^2 + n^3)$ \tnote{a} & $f < n/2$ \tnote{b} \\
    % Dolev-Strong~\cite{dolev1983authenticated} & trusted setup & $O(\kappa n^2 + n^3)$ \tnote{c} & $f < n/2$ \\
    % HotStuff~\cite{yin2019hotstuff}, Spiegelman~\cite{spiegelman2020search} \tnote{d} & trusted PKI & $O(\kappa n^2)$ & $f < n/3$ \\
     \hline
    \textbf{this paper} & threshold signature & $O(\kappa n^2)$ & $f < n/2$ \\
    \textbf{this paper} & PKI & $O(\kappa n^2)$ & $f \le (\frac 12 - \varepsilon)n$ \\ \hline
  \end{tabular} 
%   \label{tab:comparison_with_prior_work}
  \begin{tablenotes}
    \small
    \item[a] The original Dolev-Strong protocol solves BB but can be easily converted into a BA protocol with an initial round to multicast the inputs. Using a multi-signature with a list of signer identities attached, the protocol achieves $O(\kappa n^2 + n^3)$. 
    
    \item[b] Although the original Dolev-Strong BB protocol tolerates $f < n$ faults, converting it to a BA protocol decreases the fault tolerance to $f < n/2$, which is optimal for authenticated BA.        
    
    % \item[c] Using a threshold signature and a multi-signature with a list of signer identities attached, the protocol achieves $O(\kappa n^2 + n^3)$. 

  \end{tablenotes}
  \medskip
  \end{threeparttable}
  \caption{Upper bounds for worst-case communication complexity of deterministic Byzantine agreement. $\varepsilon$ is any positive constant.}
\end{table*}
\label{tab:comparison_with_prior_work}

As we can see, the above two results achieve quadratic communication with different trade-offs. The first result achieves the optimal resilience $f < n/2$ but relies on a trusted setup due to the use of threshold signature. On the other hand, the second result is in the standard PKI model, but there is a small gap in the resilience. 
Table \ref{tab:comparison_with_prior_work} compares our results to the current landscape of worst-case communication complexity of BA.

\bos{Comparing with state-of-the-art BA solutions.}
Although our primary motivation of this study is to close the gap on the worst-case communication complexity of deterministic BA, the second result in Theorem~\ref{theo:ba_without_trusted_setup} has some advantage even over state-of-the-art randomized protocols.
To the best of our knowledge, our second protocol is the first to achieve the following three properties simultaneously under the standard PKI model: (1) near-optimal resilience of $f \le (\frac 12 - \varepsilon)n$, (2) security against an 
%\rl{removed strong, since we didn't define it}
adaptive adversary, (3) expected sub-cubic (in fact, we achieve worst-case quadratic) communication complexity.
The works of Berman et al.~\cite{berman1992bit} and King-Saia~\cite{king2011breaking} achieve (sub-)quadratic communication and adaptive security but tolerate only $f < n/3$.
Abraham et al.~\cite{abraham2019synchronous,abraham2019communication} achieve (sub-)quadratic communication and adaptive security under $f \le (\frac 12 - \varepsilon)n$, but require some trusted setup assumption due to the use of threshold signature or verifiale random functions.
Tsimos et al.~\cite{tsimos2020nearly} recently achieve nearly-quadratic communication in the standard PKI model for $f \le (1 - \varepsilon)n$ (for broadcast), but it is secure only against a static adversary.

\bos{Organization.}
The rest of the paper is organized as follows.
In the rest of this section, we briefly review related work and give an overview of the techniques we use to achieve our two results. 
Section \ref{sec:preliminaries} introduces definitions, models and notations.
Section \ref{sec:recursive_ba} introduces the recursive framework to get a BA protocol with quadratic communication including the definition of GBA primitive.
Section \ref{sec:gbas} presents two GBA protocols to instantiate two BA protocols with different trade-offs to complete our results.
Finally, we discuss future directions and conclude the paper in Section \ref{sec:conclusion}.

\subsection{Technical Overview}

\bos{Abstracting the recursive framework of Berman et al.}
To obtain the results, we revisit the Berman et al.~\cite{berman1992bit} protocol. At a high level, Berman et al. is a recursive protocol: it partitions parties into two halves recursively until they reach a small instance with sufficiently few (e.g., a constant number of) participants. Since the upper bound on the fraction of faults $1/3$ is preserved in at least one of two halves, the ``correct'' half directs the entire parties to reach an agreement. If the communication except the two recursive calls is quadratic, the communication complexity of the entire protocol is also quadratic. The challenge is to prevent an ``incorrect'' run of recursive call (in a half with more than $1/3$ faults) from ruining the result. Berman et al. solve this problem with a few additional rounds of communication called ``universal exchange'' before each recursive call. It helps honest parties stick to a value when all honest parties already agree on the value, thus preventing an incorrect recursive call from changing the agreed-upon value.

Back to our setting of $f \ge n/3$, we will use the recursive framework of Berman et al.. 
However, the universal exchange step of Berman et al. relies on a quorum-intersection argument, which only works under $f < n/3$. 
To achieve BA with $f \ge n/3$, we observe that the functionality achieved by the universal exchange can be abstracted as a new primitive called graded Byzantine agreement (GBA), which we formally define in Section \ref{sec:recursive_ba}.
If we can construct a GBA with quadratic communication and plug it into the recursive framework, we will obtain a BA protocol with quadratic communication.
Thus, it remains to construct quadratic GBA.

\bos{Two constructions of Graded BA with different trade-offs.}
As the name suggests, the GBA primitive shares some similarities with graded broadcast studied in~\cite{katz2009expected,abraham2019synchronous}, but it is harder to construct due to the fact that every party has an input. This can be addressed in two ways, leading to our two constructions. 

The first method way is to resort to the (well-established) use of threshold signatures~\cite{cachin2001secure,yin2019hotstuff}. Roughly, a threshold signature condenses a quorum of $n-f = \Omega(n)$ votes into a succinct proof of the voting result.
This way, a verifiable voting result can be multicasted to all parties using quadratic total communication (linear per node).
This achieves Theorem \ref{theo:ba_exist_optimal} and requires a trusted setup for threshold signature.

Next, we try to construct a quadratic GBA without trusted setup or threshold signature scheme. 
This turns out to be much more challenging.
Na\"ively multicasting the voting result would require quadratic communication per node (cubic in total) since the voting result consists of a linear number of votes.
To get around this problem, we replace the multicast step with communication through an expander graph with constant degree.
As each party transmits the voting result to only a constant number of neighbors, the communication is kept quadratic in total even though the voting result consists of a linear number of votes.
Our key observation is that even though some of the honest parties may fail to receive or transmit the voting result (because all their neighbors are corrupted), as long as a small but linear number of honest parties transmit the voting result, the good connectivity of the expander helps prevent inconsistent decisions between honest parties.
In order to verify a linear number of honest parties actually transmit, a quorum of $n-f$ parties who claim to have transmitted should contain at least a linear number of honest parties, which results in the gap of $\epsilon n$ in the resilience in Theorem \ref{theo:ba_without_trusted_setup}.

\begin{comment}
\bos{Summary of results.}
\begin{enumerate}[topsep=2pt,itemsep=2pt]
    \item A Byzantine agreement protocol with quadratic communication and optimal resilience $f < n/2$ under a trusted PKI model;
    
    \item A Byzantine agreement protocol with quadratic communication and near optimal resilience $f \le (\frac 12 - \varepsilon)n$ without trusted setup;
\end{enumerate}
\end{comment}

\subsection{Related Work}

Byzantine Agreement was first introduced by Lamport et al.~\cite{pease1980reaching,lamport1982byzantine}. Without cryptography (i.e., the unauthenticated setting), BA can be solved if and only if $f < n/3$. Assuming a digital signature scheme with a public-key infrastructure (i.e., the authenticated setting), BA can be solved if and only if $f < n/2$.
Lamport et al. gave BA protocols for both settings, but they both require exponential communication. 
Later, polynomial communication protocols were shown in both settings. In particular, Dolev and Strong~\cite{dolev1983authenticated} showed a $O(\kappa n^3)$ communication protocol for the authenticated setting and Dolev et al.~\cite{dolev1982efficient} showed a $O(n^3 \log n)$ communication protocol for the unauthenticated setting. For the unauthenticated setting, Berman et al. further reduced the communication to $O(n^2)$,
matching a lower bound established by Dolev and Reischuk~\cite{dolev1985bounds}, which states that any deterministic protocol must incur $\Omega(n^2)$ communication complexity.
A recent work called HotStuff~\cite{yin2019hotstuff} can be modified~\cite{spiegelman2020search} to achieve $O(\kappa n^2)$ communication with $f < n/3$ for the authenticated setting. 
%Spiegelman~\cite{spiegelman2020search} improved it to get $O(\kappa tn)$ communication where $t$ is the number of actual faults; note that in the worst-case, i.e., $t = f = O(n)$, the communication will still be $O(\kappa n^2)$.

We also mention two orthogonal lines of work.
Some works known as extension protocols~\cite{cachin2005asynchronous,miller2016honey,nayak2020improved,lu2020dumbo} achieve an optimal $O(nl)$ communication complexity for sufficiently long inputs of size $l$ using the BA oracle for short inputs.
When the input size is small, e.g., $l = O(1)$, the communication complexity degenerates to that of the underlying BA oracle.
Our work provides improved oracles for these protocols.
Another line of works study randomized protocols to get expected quadratic~\cite{feldman1988optimal,katz2009expected,cachin2001secure,micali2016byzantine,abraham2018validated,abraham2019synchronous} or even sub-quadratic communication~\cite{king2011breaking,chen2016algorand,abraham2019communication}.
Naturally, they do not address the tightness of the Dolev-Resischuk lower bound for deterministic protocols. 
In contrast, even though our protocol will also have an error probability from signature schemes, our protocol is otherwise deterministic and error-free, and hence, is subject to the Dolev-Resischuk lower bound.

\section{Preliminaries} \label{sec:preliminaries}

\bos{Execution model.}
We define a protocol as an algorithm for a set of parties. There are a set of $n$ parties, of which at most $f < n$ are Byzantine faulty and behave arbitrarily. We assume $f = \Theta(n)$.  
All presented protocols are secure against $f$ adaptive corruption that can happen anytime during the protocol execution.
A party that is not faulty throughout the execution is said to be honest and faithfully execute the protocol. We use the term \emph{quorum} to mean the minimum number of all honest parties, i.e., $n-f$. 
A protocol proceeds in synchronous rounds. If an honest party sends a message at the beginning of some round, an honest recipient receives the message at the end of that round. 

We assume digital signatures and public-key infrastructure (PKI), and use $\langle x \rangle_{r}$ to denote a message $x$ signed by party $r$. 
As mentioned, one of our protocols assumes a threshold signature scheme~\cite{cachin2001secure,libert2016born}. In the threshold signature scheme, a set of signatures $\langle x \rangle_{r}$ for a message $x$ from $t$ (the threshold) distinct parties can be combined into a threshold signature for $x$ with the same length as an individual signature.
The currently known threshold signature schemes require a trusted dealer who generates all public and private keys for all parties and a group public key to verify an aggregated full signature, henceforth we call it trusted setup.
Our second protocol is in the standard PKI model and does not require any trusted setup beyond that.
In that case, each party independently generates a pair of public and private keys without any extra assumption.
As commonly done in Byzantine agreement, we abstract away the details of cryptography, namely, we assume the (threshold) signature schemes are ideal.

\bos{Complexity metrics.}
The communication complexity of a protocol is the maximum number of bits sent by all honest parties combined across all executions. 
Since all messages in our protocols are signed, we use the signature size $\kappa$ as the unit of measure for communication.
We assume the size of any input value is on the order of $\kappa$. 
The Dolev-Reichuk lower bound, however, is in terms of the number of messages. 
With no assumption on the message size, this leaves a gap of $\kappa$ in a pair of matching upper and lower bounds. 
If we further assume that every message in authenticated protocols is signed, then the bounds match.
It is an interesting open problem whether we can design an authenticated protocol that leaves most of the messages \emph{unsigned} to do better than $O(\kappa n^2)$.

\bos{Byzantine Agreement.}
In Byzantine Agreement (BA), each party has an input value, and all parties try to decide on the same value. 
The requirement of BA is defined as follows.

\begin{definition}[Byzantine Agreement (BA)]
    A Byzantine agreement protocol must satisfy the following properties. 
    \begin{enumerate}
        \item consistency. if two honest parties $r$ and $r'$ decide values $v$ and $v'$, then $v = v'$.
        \item termination. every honest party decides a value and terminates.
        \item validity. if all honest parties have the same input value, then all honest parties decide that value.
    \end{enumerate}
\end{definition}

Although our main focus of this paper is BA, we also mention a closely related problem called Byzantine broadcast (BB).
In BB, a designated sender has an input to broadcast to all parties, and all parties try to decide on the same value.
The requirement of BB is defined as follows.

\begin{definition}[Byzantine Broadcast (BB)]
    A Byzantine broadcast protocol must satisfy the following properties. 
    \begin{enumerate}
        \item consistency. same as above.
        \item termination. same as above.    
        \item validity. if the sender is honest, then all honest parties decide the sender's value.
    \end{enumerate}    
\end{definition}

It is easy to transform a BA protocol into a BB protocol preserving the same resilience and quadratic communication complexity by having an initial round for the sender to broadcast its input value before starting the BA protocol~\cite{lamport1982byzantine}.
As the Dolev-Reischuk lower bound holds for both BA and BB, our results establish the tightness of the quadratic communication complexity for BB as well (though the resilience $f < n/2$ is not optimal for BB, which is possible under any $f<n$).

\section{Recursive Framework of Byzantine Agreement with Quadratic Communication} \label{sec:recursive_ba}

This section reviews the recursive framework to construct a BA protocol with quadratic communication introduced by Berman et al. \cite{berman1992bit} for $f < n/3$, and making it works for $f < n/2$.

\bos{Dissecting Berman et al.}
In the Berman et al. protocol, parties are partitioned into two halves, and each half runs the BA protocol recursively in sequential order. 
The partition continues until we reach a BA instance with a constant number of parties, where using any inefficient BA protocol will not impact the overall complexity. 
At each recursive step, additional quadratic communication is incurred besides the two recursive BA calls.  
It is not hard to see that the overall communication complexity is quadratic.

Since the fraction of faults in the entire parties is less than $1/3$, one of two halves also has faults of less than $1/3$ and thus achieve a ``correct'' BA. However, even if the first committee is correct, the potential incorrect second BA instance may ``ruin'' the result of the first one. 
To prevent this, parties run a few rounds of preprocessing steps called ``universal exchange'' in Berman et al. before each recursive BA call. 
The universal exchange step helps parties ``stick to'' a value (ignoring the recursive BA output) if all honest parties already agree on that value.
In more detail, if the first run of recursive BA is correct and all honest parties agree on a value, the universal exchange before the second run makes sure all honest parties  stick to it and the second run cannot change the agreed-upon value. 

A tricky situation this universal exchange step needs to handle is when some honest parties stick to a value but other parties do not.
In this case, this step needs to ensure that, if any honest party sticks to a value, other parties at least input that value to the subsequent BA call.
The validity property of a correct recursive BA call will ensure agreement.

Here, the above recursive construction itself is independent of $f$, but the universal exchange step of Berman et al. relies on a quorum-intersection argument which only works under $f < n/3$.
To make the framework independent of $f$, we abstract the functionality of this step as \emph{graded Byzantine agreement} (GBA), since it is essentially the agreement version of graded broacast~\cite{feldman1988optimal,katz2009expected}.
In the rest of this section, we formally define the GBA primitive and construct a BA protocol using a GBA protocol as a black-box and prove its correctness.

\subsection{Graded Byzantine Agreement}

We introduce a primitive we call graded Byzantine agreement (GBA). In GBA, each party $r$ has an input value, and outputs a tuple $(v, g)$ where $v$ is the output value and $g \in \{0, 1\}$ is a grade bit.

\begin{definition}[Graded Byzantine Agreement (GBA)]
    A Graded Byzantine agreement protocol must satisfy the following properties. 
    \begin{enumerate}
        \item consistency. if an honest party outputs $(v, 1)$, then all honest parties output $(v, *)$.
        \item validity. if all honest parties have the same input value $v$, then all honest parties output $(v, 1)$
        \item termination. every honest party outputs and terminates.
    \end{enumerate}
\end{definition}

The ``stick to'' nature is expressed by the grade bit $g$. The consistency property requires that if an honest party sticks to a value $v$, i.e., output $v$ with $g = 1$, then all honest parties output the same value $v$. The validity property states that if all honest parties have the same input value $v$, they all stick to the value. These two properties capture what the universal exchange step needs to achieve explained at an intuitive level.

\subsection{Recursive Construction of Byzantine Agreement}

\begin{figure*}[t!]
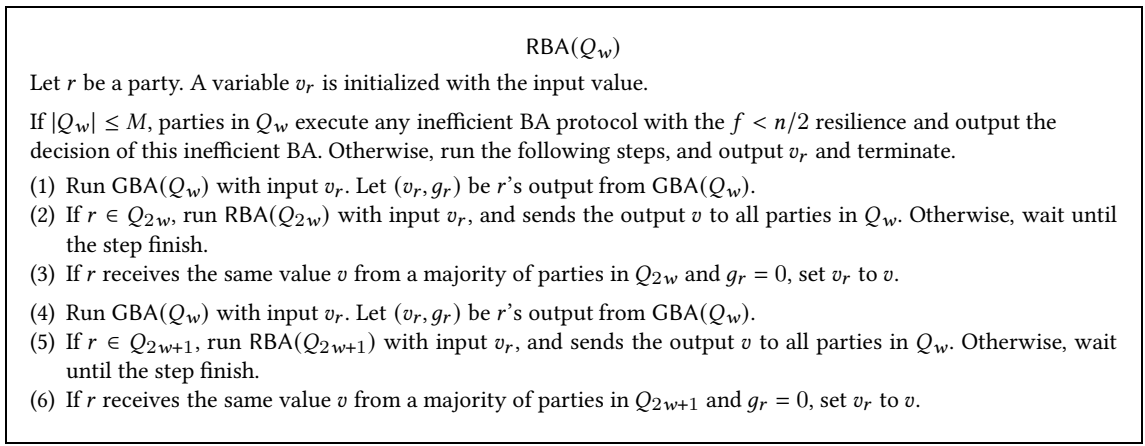

% \begin{center}
% \begin{itembox}[]{$\recba(Q_{w})$}
    \begin{framed}
    \begin{center} $\recba(Q_{w})$ \end{center}
    \begin{flushleft}
    Let $r$ be a party. A variable $v_{r}$ is initialized with the input value.
    
    \medskip
    If $|Q_{w}| \le M$, parties in $Q_{w}$ execute any inefficient BA protocol with the $f < n/2$ resilience
    and output the decision of this inefficient BA. Otherwise, run the following steps, and output $v_{r}$ and terminate.
    \end{flushleft}
    
    \begin{enumerate}[leftmargin=*,itemsep=0pt]
        \item Run $\GBA(Q_{w})$ with input $v_{r}$. Let $(v_{r}, g_{r})$ be $r$'s output from $\GBA(Q_{w})$.
        
        \item If $r \in Q_{2w}$, run $\recba(Q_{2w})$ with input $v_{r}$, and sends the output $v$ to all parties in $Q_{w}$. Otherwise, wait until the step finish.
        
        \item If $r$ receives the same value $v$ from a majority of parties in $Q_{2w}$ and $g_{r} = 0$, set $v_{r}$ to $v$. 

        \medskip
        
        \item Run $\GBA(Q_{w})$ with input $v_{r}$. Let $(v_{r}, g_{r})$ be $r$'s output from $\GBA(Q_{w})$.
        
        \item If $r \in Q_{2w+1}$, run $\recba(Q_{2w+1})$ with input $v_{r}$, and sends the output $v$ to all parties in $Q_{w}$. Otherwise, wait until the step finish.
        
        \item If $r$ receives the same value $v$ from a majority of parties in $Q_{2w+1}$ and $g_{r} = 0$, set $v_{r}$ to $v$. 
    \end{enumerate}
    \end{framed}
    \caption{Byzantine Agreement with $O(\kappa n^2)$ communication and $f < \frac{n}{2}$.}
    \label{fig:ba_rec}
% \end{center}
\end{figure*}

Next, we present the recursive BA protocol $\recba$ in Figure \ref{fig:ba_rec}. Let $Q_{w}$ denote a set of parties that run a BA protocol. Since the protocol is recursive, the set $Q_{w}$ is also defined recursively. $Q_{1}$ is a set of all $n$ parties. $Q_{2w}$ is the first $\lceil |Q_{w}|/2 \rceil$ parties in $Q_{w}$, and $Q_{2w+1}$ is the remaining $\lfloor |Q_{w}|/2 \rfloor$ parties.
All parties start by running $\recba(Q_1)$ at the beginning.

If the size of the $\recba$ instance gets below a constant, denoted as $M$ in the figure, parties can run any inefficient BA protocol with cubic or even higher communication complexity but with the desired resilience up to $f < n/2$.
There are many such constructions in the literature~\cite{lamport1982byzantine,dolev1983authenticated,katz2009expected,abraham2019synchronous}; we do not describe these protocols. 
Otherwise, parties run two instances of $\recba$ recursively to further reduce the instance size.
Before each recursive call, they run a given GBA protocol denoted $\GBA$.
The grade bit output $g_r$ of the $\GBA$ determines if a party $r$ ``sticks to'' the $\GBA$ output or adopts the recursive $\recba$ output. 

\bos{Correctness of the Protocol.}
We prove the correctness of $\recba$ for $f < n/2$ assuming the given GBA protocol $\GBA$ also tolerates $f < n/2$. The proof is easily extended for $f \le (\frac 12 - \varepsilon)n$.
Below, minority faults within a set of parties $Q$ mean at most $\lfloor (|Q|-1)/2 \rfloor$ faults.
    
\begin{lemma}
    $\recba$ solves BA in the presence of minority faults.
\end{lemma}
\begin{proof}
    Termination is obvious. 
    The proof for validity is also easy. If all honest parties have the same input value $v_{r} = v$, then due to the validity of $\GBA$, all honest parties output $(v,1)$ in step-1.
    Thus, they do not change $v_{r}$ at step-3 and input $v_r=v$ into the $\GBA$ of step-4.
    Again due to the validity of $\GBA$, all honest parties output $(v,1)$ in step-4, do not change $v_{r}$ at step-6, and all output $v$.
    
    Next, we prove consistency.
    When $|Q| \le M$, the correctness of $\recba$ reduces to the correctness of the given inefficient BA. We just need to prove for the recursive step. Specifically, we will prove that $\recba$ solves BA under $n$ parties with minority faults, if $\recba$ solves BA under $<n$ parties with minority faults.
    
    Consider $\recba(Q_w)$.
    Since $Q_{w}$ has minority faults, at least one of the two halves $Q_{2w}$ and $Q_{2w+1}$ has minority faults. 
    Let us first consider the case where $Q_{2w}$ has minority faults. Here, there are two situations with regard to the result of step-1: (i) all honest parties in $Q_{w}$ set $g_{r}$ to 0, or (ii) at least an honest party in $Q_{w}$ sets $g_{r}$ to 1. 
    
    In the first situation, all honest parties will set $v_r$ to the majority output of step-2. By the consistency and termination of $\recba(Q_{2w})$, all honest parties in $Q_{w}$ receive the same value $v$ from honest parties in $Q_{2w}$ (which constitute a majority in $Q_{2w}$ ). Thus, all honest parties in $Q_{w}$ set $v_{r}$ to $v$ in step-3. 
    
    In the second situation, since some honest party sets $g_{r}$ to 1 in step-1, then by the consistency of $\GBA$, all honest parties in $Q_{w}$ set $v_{r}$ to the same value $v$ at the end of step-1. By the validity of $\recba(Q_{2w})$, all honest parties in $Q_{2w}$ output $v$, so all honest parties in $Q_w$ receive $v$ from a majority of parties in $Q_{2w}$.
    Thus, all honest parties in $Q_{w}$ set $v_{r}$ to $v$ in step-3.
    
    Therefore, in both situations, all honest parties in $Q_{w}$ have the same value $v_{r} = v$ at the beginning of step-4. Then, by the validity of $\GBA$, all honest parties in $Q_{w}$ set $(v_{r}, g_{r})$ to $(v, 1)$ in step-4, so will not change their $v_{r}$ in step-6 and all output the same value $v$. 
    
    The other case where $Q_{2w+1}$ has minority faults can be proved similarly.
    No matter which of the two situations holds at step-4 (all have $g_r=0$ or some have $g_r=1$),  
    all honest parities in $Q_{w}$ have the same value $v_{r} = v$ at the end of step-6 and output the same value $v$. 
    Therefore, regardless of whether $Q_{2w}$ or $Q_{2w+1}$ has minority faults, consistency holds. 
\end{proof}

With some foresight, we will construct $\GBA$ with quadratic communication in the later section.
This will give $\recba$ with quadratic communication in total.

\begin{lemma}[Communication Complexity]
    If the communication complexity of $\GBA$ is $O(\kappa n^2)$, then the communication complexity of $\recba$ is $O(\kappa n^2)$.
\end{lemma}
\begin{proof}
    The communication complexity of $\recba$ is given as a recurrence below. Let $s$ be the number of parties in an RBA instance.
    \begin{equation*}
    C(s) = \begin{cases}
                O(\kappa) ~~~~ (\mbox{if } s \le M) \nonumber \\
                C(\lfloor s/2 \rfloor) + C(\lceil s/2 \rceil) + O(\kappa s^2) ~~~~ (\mbox{otherwise})
    \end{cases}
    \end{equation*}
    For any $n$, the depth of the recursion $k$ satisfies $2^{k-1}M \le n \le 2^k M$. Hence, $C(n) \le 2^k O(\kappa)$ + $\sum_{i=0}^{k} 2^i O(\kappa (n/2^i)^2) = O(\kappa n^2)$.
\end{proof}

% \bos{Discussion.}
% A drawback of our protocol is related to the need for threshold signature in GBA. 
% As each recursive RBA instance has a different number of parties and threshold, each party needs to have trusted threshold key setups for a logarithmic number of instances.
% It is an interesting open question whether we can construct quadratic BA with less or no trusted setup.

%As the main idea of recursive construction by Berman et al. itself does not need a signature aggregation, GBA with quadratic communication without threshold signing scheme makes the communication of the entire protocol quadratic without trusted setup. 

% Finally, as each party cannot decide a value until both two recursions and $\GBA$ finish, $\recba$ cannot easily achieve ``early stopping'' property \cite{dolev1990early}.

\section{Graded Byzantine Agreement with Different Trade-offs}
\label{sec:gbas}

This section presents two constructions of GBA protocols with different trade-offs to instantiate two BA protocols from the recursive framework in the previous section and complete the proof of Theorem \ref{theo:ba_exist_optimal} and \ref{theo:ba_without_trusted_setup}.

\subsection{Graded Byzantine Agreement with Threshold Signature Scheme}
We first present a GBA protocol (denoted $\OptGBA$) with quadratic communication and $f < n/2$ assuming a threshold signature scheme, which complete the proof of Theorem \ref{theo:ba_exist_optimal}.
We describe $\OptGBA$ in Figure~\ref{fig:graded_ba}. The parameter $Q$ is a set of parties that participate in the protocol. Let $n = |Q|$.

\bos{Intuitive overviews.}
The construction is inspired by a few recent work on synchronous BB and BFT protocols~\cite{abraham2019synchronous,synchotstuff,optimalsmr2020}. Rounds 1--3 form a set of $n-f$ $\FirstVote$ (vote1-certificate) for the same value $v$, denoted $\Cert^{1}(v)$. Here, if an honest party votes for a value $v$ in round 3, it must have received and multicast $n-f$ $\Echo$ (echo-certificate) for $v$, denoted $\EchoCert(v)$ in round 2. Moreover, if a party receives a conflicting echo-certificate $\EchoCert(v')$ by the end of round 2, it does not vote in round 3. Therefore, rounds 1 and 2  prevent conflicting vote1-certificates from being created. 

Round 4 forms a set of $n-f$ $\SecondVote$ (vote2-certificate) for a value $v$, denoted $\Cert^{2}(v)$. 
If a party receives a vote1-certificate $\Cert^{1}(v)$ by the end of round 3, it sends $\SecondVote$ for a value $v$ (along with $\Cert^{1}(v)$) in round 4. Therefore, if a vote2-certificate $\Cert^{2}(v)$ is formed, all honest parties can receive a vote1-certificate $\Cert^{1}(v)$.

Finally, a party outputs a value $v$ if it receives a vote1-certificate $\Cert^{1}(v)$, and it further sets the grade bit $g$ to 1 if it also receives a vote2-certificate $\Cert^{2}(v)$. Consistency follows from the properties above. Moreover, if all honest parties have the same input value $v$, all honest parties (at least $n-f$) receive both $\Cert^{1}(v)$ and $\Cert^{2}(v)$ and output $(v, 1)$, hence validity also holds. 

\begin{figure*}[t!]
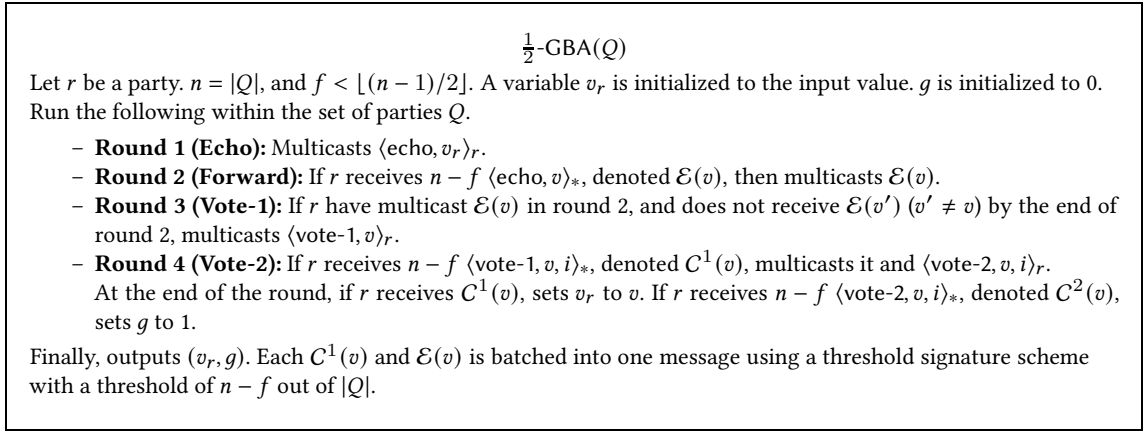

% \begin{center}
    % \begin{itembox}[]{$\OptGBA(Q, v_{r}) \rightarrow (v_{out}, g)$ }
    \begin{framed}
    \begin{center} $\OptGBA(Q)$ \end{center}
    \begin{flushleft}
    Let $r$ be a party. $n = |Q|$, and $f < \lfloor (n-1)/2 \rfloor$. A variable $v_{r}$ is initialized to the input value. $g$ is initialized to 0. Run the following within the set of parties $Q$. 
    \end{flushleft}
    
    \begin{enumerate}
        \item[--] \textbf{Round 1 (Echo):} Multicasts $\langle \mbox{\sf echo},v_{r} \rangle_{r}$.
        
        \item[--] \textbf{Round 2 (Forward):} If $r$ receives $n-f$ $\langle \mbox{\sf echo},v \rangle_{*}$, denoted $\qer{}{v}$, then multicasts $\qer{}{v}$.
        
        \item[--] \textbf{Round 3 (Vote-1):} If $r$ have multicast $\qer{}{v}$ in round 2, and does not receive $\qer{}{v'}$ ($v' \neq v$) by the end of round 2, multicasts $\langle \mbox{\sf vote-1},v \rangle_{r}$. 
        
        \item[--] \textbf{Round 4 (Vote-2):} If $r$ receives $n-f$ $\langle \mbox{\sf vote-1},v,i \rangle_{*}$, denoted $\fcer{}{v}$, multicasts it and $\langle \mbox{\sf vote-2},v,i \rangle_{r}$. 
        
        At the end of the round, if $r$ receives $\fcer{}{v}$, sets $v_{r}$ to $v$. If $r$ receives $n-f$ $\langle \mbox{\sf vote-2},v,i \rangle_{*}$, denoted $\scer{}{v}$, sets $g$ to 1. 
    \end{enumerate}
    
    \begin{flushleft}
    Finally, outputs $(v_{r}, g)$. Each $\fcer{}{v}$ and $\qer{}{v}$ is batched into one message using a threshold signature scheme with a threshold of $n-f$ out of $|Q|$.
    \end{flushleft}
    \end{framed}
    % \end{itembox}
    \caption{Graded Byzantine agreement with $f < n/2$ with a threshold signature scheme.}
    \label{fig:graded_ba}
% \end{center}
\end{figure*}

\bos{Correctness of the protocol.}
We prove the correctness of $\OptGBA$ assuming $f < n/2$.
The termination of $\OptGBA$ is trivial, and thus we prove the consistency and validity. 

\begin{lemma} \label{lemma:certified_without_equivocation_ba_cubic}
    If $\fcer{}{v}$ and $\fcer{}{v'}$ are both created, then $v = v'$. 
\end{lemma}
\begin{proof}
    Suppose $\fcer{}{v}$ is created, then at least an honest party $r$ must have multicast $\SecondVote$ for $v$ in round 3. That implies $r$ received $\qer{}{v}$ and multicast it in round 2. Then, all honest parties must have received $\qer{}{v}$ by round 3, and all honest parties could not have multicast $\SecondVote$ for $v' \neq v$. Therefore, $\fcer{}{v'}$ cannot be created unless $v' = v$.
\end{proof}

\begin{lemma}[Consistency]
    If an honest party outputs $(v, 1)$, then all honest parties output $(v, *)$
\end{lemma}
\begin{proof}
    Suppose an honest party outputs $(v, 1)$, then it must have received $\scer{}{v}$ for a value $v$ by the end of round 4. Then, at least one honest party must have multicast $\fcer{}{v}$ in round 4, and all honest parties must have received it by the end of round 4. Since there is not $\fcer{}{v'}$ for a different value $v'$ by Lemma \ref{lemma:certified_without_equivocation_ba_cubic}, all honest parties set $v_r$ to $v$ at the end of round 4 and thus output $v$.
\end{proof}

\begin{lemma}[Validity]
    If all honest parties have the same input value $v$, then all honest parties output $(v, 1)$
\end{lemma}
\begin{proof}
    If all honest parties have the same input value $v$, they all multicast $\sig{\Echo, v}$ in round 1, and thus $\qer{}{v}$ should be formed and $\qer{}{v'}$ for $v' \neq v$ cannot be formed. In the same way, all honest parties multicast $\sig{\FirstVote, v}$ in round 3 and $\sig{\SecondVote, v}$ in round 4. Therefore, $\fcer{}{v}$ and $\scer{}{v}$ should be formed and $\fcer{}{v'}$ and $\scer{}{v'}$ for $v' \neq v$ cannot be formed. Thus, all honest parties output $(v, 1)$.
\end{proof}

\subsection{Graded Byzantine Agreement without Threshold Signature Scheme}

Next, we present a GBA protocol (denoted $\EpsGBA$) 
with quadratic communication and $f \le (\frac 12 - \varepsilon)n$ for any positive constant $\varepsilon$ without threshold signature scheme.
We describe $\EpsGBA$ in Figure \ref{fig:graded_ba_without_tsig}.

% As mentioned, the need for a trusted setup in the BA protocol in the previous section is due to the use of a threshold signature scheme.
% However, the threshold signature scheme is used only in the GBA protocol $\OptGBA$, and thus it is enough to construct a GBA protocol without trusted setup, and combining with the recursive framework implies BA protocol with the same resilience without trusted setup.
% We will show in Figure \ref{fig:graded_ba_without_tsig} a GBA protocol $\EpsGBA$ without threshold signature scheme.

\begin{figure*}[t!]
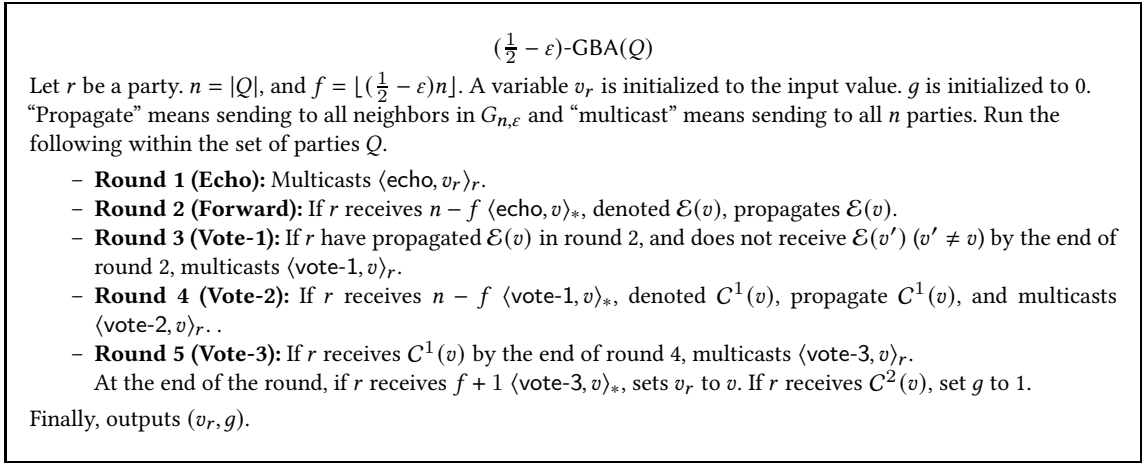

    \begin{framed}
    \begin{center} $\EpsGBA(Q)$ \end{center}
    \begin{flushleft}
    Let $r$ be a party. $n = |Q|$, and $f = \lfloor (\frac 12 - \varepsilon)n \rfloor$. A variable $v_{r}$ is initialized to the input value. $g$ is initialized to 0. ``Propagate'' means sending to all neighbors in $G_{n,\varepsilon}$ and ``multicast'' means sending to all $n$ parties. Run the following within the set of parties $Q$. \end{flushleft}
    
    \begin{enumerate}
        \item[--] \textbf{Round 1 (Echo):} Multicasts $\sig{\Echo, v_{r}}_{r}$.
        
        \item[--] \textbf{Round 2 (Forward):} If $r$ receives $n-f$ $\sig{\Echo, v}_{*}$, denoted $\EchoCert(v)$, propagates $\EchoCert(v)$.
        
        \item[--] \textbf{Round 3 (Vote-1):} If $r$ have propagated $\EchoCert(v)$ in round 2, and does not receive $\EchoCert(v')$ ($v' \neq v$) by the end of round 2, multicasts $\sig{\FirstVote, v}_{r}$.
        
        \item[--] \textbf{Round 4 (Vote-2):} If $r$ receives $n-f$ $\sig{\FirstVote, v}_{*}$, denoted $\Cert^{1}(v)$, propagate $\Cert^{1}(v)$, and multicasts $\sig{\SecondVote, v}_{r}$. .
        
        \item[--] \textbf{Round 5 (Vote-3):} If $r$ receives $\Cert^{1}(v)$ by the end of round 4, multicasts $\sig{\ThirdVote, v}_{r}$.
        
        At the end of the round, if $r$ receives $f+1$ $\sig{\ThirdVote, v}_{*}$, sets $v_{r}$ to $v$. 
        If $r$ receives $\Cert^{2}(v)$, set $g$ to 1.
    \end{enumerate}

    \begin{flushleft}
    Finally, outputs $(v_{r}, g)$.    
    \end{flushleft}

    \end{framed}
    \caption{Graded Byzantine agreement with $f \le (\frac 12 - \varepsilon)n$ without threshold signature scheme.}
    \label{fig:graded_ba_without_tsig}
\end{figure*}

\bos{Intuitive overview.}
The main motivation of $\EpsGBA$ is to remove the use of threshold signature. 
Thus, let us first review why threshold signature scheme is necessary in the GBA protocol $\OptGBA$ in the previous section.
The threshold signature scheme is used to aggregate a set of $n-f$ signatures (quorum certificate) before multicasting it specifically in two parts: (1) aggregating echo-certificate $\EchoCert(v)$ in round 2, and (2) aggregating vote1-certificate $\Cert^{1}(v)$ in round 4.
If these are not aggregated, each party needs to multicast linear-sized certificates, leading to cubic communication in total.

Therefore, to remove aggregation while keeping the communication quadratic, we need to remove multicast.
However, multicasting a quorum certificate in round 2 and 4 is the key to consistency.
Specifically, multicasting an echo-certificate $\EchoCert(v)$ in round 2 helps honest parties detect a conflicting echo-certificate $\EchoCert(v')$ and guarantees the unique existence of a vote1-certificate $\Cert^{1}(v)$, which allows honest parties to decide the value $v$ safely.
Furthermore, multicasting a vote1-certificate $\Cert^{1}(v)$ in round 4 helps notify all honest parties of the existence of $\Cert^{1}(v)$, which allows the party to decide the value $v$ with confidence, i.e., grade bit $g = 1$.

Our key new technique is to replace the multicast with a more efficient yet robust dissemination of quorum certificates through a predetermined expander graph with a constant degree.

\begin{definition}[Expander] \label{defi:expander}
    Let $\alpha$ and $\beta$ be constants satisfying $0<\alpha<\beta<1$. An $(n, \alpha, \beta)$-expander is a graph of $n$ vertices such that, for any set $S$ of $\alpha n$ vertices, the number of neighbors of $S$ is more than $\beta n$.
\end{definition}

It is well-known that for any $n$ and $0<\alpha<\beta<1$, $(n, \alpha, \beta)$-expanders with constant degrees exist. 
For our purpose, we need an $(n,2\varepsilon,1-2\varepsilon)$-expander; in other words, we set $\alpha=2\varepsilon$ and $\beta=1-2\varepsilon$.
Henceforth, we write an $(n,2\varepsilon,1-2\varepsilon)$-expander as $G_{n,\varepsilon}$.
For completeness, we show in Appendix \ref{apd:expander} that for all positive constant $\varepsilon$ and for all $n$, the required expander $G_{n,\varepsilon}$ always exists.
% It is important to note that even if we use a randomized expander construction, the protocol is still deterministic, because all the randomization happens in the offline ``protocol design'' phase. Once a required expander $G_{n,\varepsilon}$ is found, it can be harded-coded into the protocol. 
% Nonetheless, deterministic expander constructions also exist, e.g., Cayley graphs and Ramanujan graphs~\cite{alon1995tools}.

Instead of sending a quorum certificate to all other parties, a party propagates it to the constant set of neighbors in $G_{n,\varepsilon}$.
Therefore, the total number of messages is reduced from quadratic to linear, and thus the total communication is kept quadratic even with the message containing a linear number of signatures. 
Our key observation is that although it is impossible to propagate a message to everyone through a constant degree expander as it is not fully connected and a linear number of parties can be Byzantine, it is sufficient to maintain consistent decisions among honest parties.

In more detail, in round 3, each party multicasts $\FirstVote$ for a value $v$ only if it propagated an echo-certificate $\EchoCert(v)$ in round 2 and it does not receive a conflicting echo-certificate $\EchoCert(v')$.
If a vote1-certificate $\Cert^{1}(v)$ forms, at least $n-2f = 2\varepsilon n$ are honest.
They must have propagated $\EchoCert(v)$ and it will be received by more than $(1-2\varepsilon)n = 2f$ parties.
Out of these, at least $f+1$ are honest and will not vote for a conflicting value.
This guarantee the unique existence of vote1-certificate $\Cert^{1}(v)$.

Confirming the existence of a vote1-certificate is trickier as we cannot afford multicasts to notify all parties. We achieve this in two steps.
In round 4, after propagating $\Cert^{1}(v)$, the party multicast $\SecondVote$ for $v$.
If a vote2-certificate $\Cert^{2}(v)$ forms, due to the expansion property, at least $f+1$ honest parties receive $\Cert^{1}(v)$ by the end of round 4. Then, in round 5, if a party receives $\Cert^{1}(v)$, it multicast $\ThirdVote$ message for $v$.
As at least $f+1$ honest parties receives $\Cert^{1}(v)$, all honest parties can receive $f+1$ $\ThirdVote$ message for $v$, which works as a succinct proof of existence of $\Cert^{1}(v)$.
This allows all honest parties to confirm the existence of a vote1-certificate.

\bos{Correctness of the protocol.}
We prove the correctness of $\EpsGBA$ assuming $f \le (\frac 12 - \varepsilon)n$ for any positive constant $\varepsilon$.
The termination of $\EpsGBA$ is trivial, and thus we prove the consistency and validity. 

\begin{lemma} \label{lemma:no_conflicting_first_cert}
    If $\Cert^{1}(v)$ and $\Cert^{1}(v')$ are both created, then $v = v'$. 
\end{lemma}
\begin{proof}
    Suppose $\Cert^{1}(v)$ is created, then at least $2\varepsilon n$ honest parties must have propagated $\EchoCert(v)$ in round 2. 
    Then, due to the expansion property of $G_{n,\varepsilon}$, more than $2f$ parties, out of which at least $f+1$ honest parties must have received $\EchoCert(v)$ by the end of round 2, and do not send $\sig{\FirstVote, v'}_{*}$ for a different value $v' \neq v$ in round 3.
    Therefore, $\Cert^{1}(v')$ cannot be created unless $v' = v$.
\end{proof}

\begin{lemma}[Consistency]
    If an honest party outputs $(v, 1)$, then all honest parties output $(v, *)$
\end{lemma}
\begin{proof}
    Suppose an honest party outputs $(v, 1)$, then it must have received $\Cert^{2}(v)$ for a value $v$ by the end of round 5.
    Then, at least $2\varepsilon n$ honest parties must have propagated $\Cert^{1}(v)$ in round 4. 
    Due to the expansion property of $G_{n,\varepsilon}$, more than $2f$ parties, out of which at least $f+1$ honest parties must have received $\Cert^{1}(v)$ by the end of round 4, and multicast $\sig{\ThirdVote, v}_{*}$ in round 5.
    Thus, all honest parties must have received $f+1$ $\sig{\ThirdVote, v}_{*}$ by the end of round 5.
    Here, as $\Cert^{1}(v')$ for a different value $v' \neq v$ cannot form by Lemma \ref{lemma:no_conflicting_first_cert}, honest parties could not have multicast $\sig{\ThirdVote, v'}_{*}$.
    Therefore, all honest party could not have received $f+1$ $\sig{\ThirdVote, v'}_{*}$, and thus output $v$.
\end{proof}

\begin{lemma}[Validity]
    If all honest parties have the same input value $v$, then all honest parties output $(v, 1)$
\end{lemma}
\begin{proof}
    If all honest parties have the same input value $v$, they all multicast $\sig{\Echo, v}$ in round 1, and thus $\EchoCert(v)$ must form and $\EchoCert(v')$ for $v' \neq v$ cannot form. 
    Then, all honest parties multicast $\sig{\FirstVote, v}$ in round 3, propagate $\Cert^{1}(v)$ and multicast $\sig{\SecondVote, v}$ in round 4, and $\sig{\ThirdVote, v}$ in round 5.
    Therefore, all honest parties receive both $\Cert^{2}(v)$ and $f+1$ $\sig{\ThirdVote, v}_{*}$, and output $(v, 1)$.
\end{proof}

\section{Conclusion} \label{sec:conclusion}
In this paper, we provided two results: (1) a BA protocol with quadratic communication with optimal resilience $f < n/2$ with a trusted setup, and (2) a BA protocol with quadratic communication with near optimal resilience $f \le (\frac 12 - \varepsilon)n$ without trusted setup. Even with our new results, some gaps in worst-case communication complexity remain. For example, can we can quadratic BA under a standard PKI model with $(\frac 12 - \varepsilon)n < f < n/2$, or quadratic BB with $f \ge n/2$ even with a trusted setup? 
These are intriguing open questions for future work.

\begin{acks}
We would like to thank Zhuolun Xiang for helpful feedback.
\end{acks}

% use ACM-Reference-Format for the references
\bibliographystyle{ACM-Reference-Format}
\bibliography{main}

\appendix

\section{Expander} \label{apd:expander}
We show that an expander $G_{\varepsilon}$ in the Definition \ref{defi:expander} exists for all positive constant $\varepsilon$. We use $\Gamma(V, G)$ to denote a set of all neighbors of $V$ in a graph $G$.

    \begin{theorem}[Existence of Expander] \label{theo:expander}
        For all positive integer $n$ and positive constant $\varepsilon$, there exists an expander $G_{n,\varepsilon}$.
    \end{theorem}
    \begin{proof}
        Let $c = 2\varepsilon$. Consider a random $d$ degree graph $G$ taking the union of random $d$ perfect matchings (if $n$ is odd, the first party has two links). In each perfect matching $P$, for any set of $c n$ parties (say $S$), and any set of $(1-c)n$ parties (say $T$), the probability that $\Gamma(S,P) \subseteq T$ is bounded above by,
        \[\Pr[\Gamma(S,P) \subseteq T] \le  (\frac{(1-c)n}{n})^{\frac{c n}{2}} = (1-c)^{\frac{c n}{2}}.\]     
        
        Thus, the probability that any set of $c n$ parties does not expand in the graph, i.e., $|\Gamma(S,G)| \le (1-c)n$ for any $S$, is bounded above by,
        \begin{eqnarray}
            & & \binom{n}{c n}\binom{n}{(1-c)n} (1-c)^{\frac{c d n}{2}} \nonumber \\
            & \le & (\frac{e}{c})^{c n}(\frac{e}{1-c})^{(1-c)n}(1-c)^{\frac{c d n}{2}} \nonumber \\
            & \le & (e (\frac{1}{c})^{c}) (\frac{1}{1-c})^{1+c(\frac{d}{2}-1)})^{n} \nonumber
        \end{eqnarray}
        
        For a sufficiently large constant $d$, the above probability is smaller than 1 (in fact, exponentially small in $n$). This means there is non-zero (in fact, overwhelmingly large) probability that a randomly chosen graph is an expander. Thus, $G_{n,\varepsilon}$ exists.
    \end{proof}

\end{document}